%% file: main.tex
\documentclass[11pt]{article}
\usepackage{amsmath,amssymb,bbm,amsthm}
\usepackage{fullpage}
\usepackage{thm-restate,color,xcolor,xspace}
\definecolor{ForestGreen}{rgb}{0.1333,0.5451,0.1333}
\definecolor{DarkRed}{rgb}{0.80,0,0}
\definecolor{Red}{rgb}{1,0,0}
\usepackage[linktocpage=true,
pagebackref=true,colorlinks,
linkcolor=DarkRed,citecolor=ForestGreen,
bookmarks,bookmarksopen,bookmarksnumbered]
{hyperref}

\usepackage{cleveref}
\usepackage{graphicx}
\usepackage{algorithm,algorithmic}

\newtheorem{theorem}{Theorem}[section]
\newtheorem{lemma}[theorem]{Lemma}

\newtheorem{fact}[theorem]{Fact}

\usepackage{thmtools}

\begin{document}

\title{Local Sherman's Algorithm for Multi-commodity Flow}
\author{
Jason Li\thanks{Carnegie Mellon University. email: \texttt{jmli@cs.cmu.edu}}
\and
Thatchaphol Saranurak\thanks{University of Michigan. email: \texttt{thsa@umich.edu}. Supported by NSF Grant CCF-2238138. Partially funded by the Ministry of Education and Science of Bulgaria's support for INSAIT, Sofia University ``St.~Kliment Ohridski'' as part of the Bulgarian National Roadmap for Research Infrastructure.}
}
\date{\today}
\maketitle

\begin{abstract}
We give the first local algorithm for computing multi-commodity flow and apply it to obtain a $(1+\epsilon)$-approximate algorithm for computing a $k$-commodity flow on an expander with $m$ edges in $(m+\epsilon^{-3}k^3D)n^{o(1)}$ time, where $D$ is the total demand. This is the first $(1+\epsilon)$-approximate algorithm that breaks the $km$ multi-commodity flow barrier, albeit only on expanders. All previous algorithms either require $\Omega(km)$ time \cite{Sherman2017area} or a big constant approximation \cite{haeupler2024low}.

Our approach is by localizing Sherman's flow algorithm \cite{Sherman13} when put into the Multiplicative Weight Update (MWU) framework. We show that, on each round of MWU, the oracle could instead work with the \emph{rounded weights} where all polynomially small weights are rounded to zero. Since there are only few large weights, one can implement the oracle call with respect to the rounded weights in sublinear time. This insight is generic and may be of independent interest. 
\end{abstract}

\input{body}

\bibliographystyle{alpha}
\bibliography{ref}

\end{document}

%% file: body.tex
\section{Introduction}

The maximum flow problem is among the most central problems in combinatorial optimization. A long line of work~\cite{ford1956maximal,Dinitz70,GoldbergR98,LiuS20} recently culminated in almost-linear time algorithms to solve maximum flow to \emph{high accuracy}~\cite{ChenKLPGS22,van2023deterministic}, i.e., an $(1+\epsilon)$-approximation in time proportional to $\textup{polylog}(1/\epsilon)$. For \emph{low-accuracy} solvers, i.e., time proportional to $\textup{poly}(1/\epsilon)$, a more recent series of developments~\cite{christiano2011electrical,Sherman13} culminated in the $\tilde{O}(m/\epsilon)$ time algorithm of Sherman~\cite{Sherman2017area} for undirected graphs. Overall, the progress in the last decade has been astronomical in both accuracy regimes, and has greatly advanced the community's understanding of this fundamental problem.

The (concurrent) \emph{multi-commodity} flow problem, on the other hand, is much less understood. Solving even $2$-commodity flow on directed graphs to high accuracy is as hard as solving general LPs~\cite{ding2022two}, another major open problem in algorithm design. For dense graphs, algorithms faster than solving LPs are known~\cite{van2023faster}. In the low-accuracy regime, Sherman's algorithm for undirected graphs directly translates to $\tilde{O}(km/\epsilon)$-time for $k$-commodity flow, which is already super-linear for moderate values of $k$. However, a dependency of $km$ is difficult to avoid, as each of the $k$ flows may use a linear number of edges, leading to a representation size of $\Omega(km)$ for even outputting the flow. In the constant-approximation setting, a sparse \emph{implicit} representation of multi-commodity flow was recently developed~\cite{haeupler2024low}, leading to $\tilde{O}((k+m)n^\delta)$ time algorithms with approximation factor $\exp(\textup{poly}(1/\delta))$ for any constant $\delta>0$. It is unclear if we can hope to improve the approximation factor to $(1+\epsilon)$.

Another orthogonal line of work focuses on \emph{local} flow algorithms on \emph{unit-capacity} graphs. A local algorithm runs in time \emph{sublinear} in the input size, i.e., without reading the entire graph, and the unit-capacity setting makes local algorithms possible. During the last decade, researchers have \emph{localized} many single-commodity flow algorithms including (binary) blocking flow algorithms \cite{OrecchiaZ14,NanongkaiSY19}, augmenting path algorithms \cite{ChechikHILP17,ForsterNYSY20}, push-relabel algorithms \cite{HenzingerRW17,SaranurakW19,chen2024parallel}, and length-constrained flow algorithms \cite{HaeuplerLS24}. These local algorithms found impressive applications in \emph{dynamic} graph data structures and fast algorithms for expander decomposition and graph connectivity. All previous local algorithms except \cite{HaeuplerLS24}, however, iteratively update the flow using the current residual graph. Thus, it is difficult to generalize to the multi-commodity setting where the different commodities may interact with each other. (While it is conceivable that the algorithm of \cite{HaeuplerLS24} could be adapted to the multi-commodity setting, this algorithm is very involved.)
Ultimately, there was no local algorithm for multi-commodity flow prior to our work. 

In this work, we give the first local algorithm for multi-commodity flow that either reports that the given demand is infeasible or feasibly routes this demand ``almost completely'' on each vertex.  
\begin{theorem}
[Informal version of \Cref{thm:main-multi})] There is an algorithm that, given a unit-capacity undirected graph $G$, a $k$-commodity demand $b$, and error parameter $\epsilon$, in $O(\epsilon^{-2}\log n(\|b\|_{0}+\epsilon^{-1}\|b\|_{1}))$ time outputs either a certificate that the demand is infeasible or a feasible $k$-commodity flow whose residual demand on each vertex $v$ for each commodity is at most $\epsilon\deg(v)$.
\end{theorem}

Then, we apply our local algorithm to solve $(1+\epsilon)$-approximate $k$-commodity flow on a unit-capacity expander in time $\approx m+\epsilon^{-3}k^3D$, where $D$ is the total demand.
This is the first $(1+\epsilon)$-approximate $k$-commodity flow algorithm on expanders that breaks the $km$ multi-commodity flow barrier (for moderate values of $k$ and $D$).

\begin{restatable}{theorem}{Expander}\label{thm:expander-main}
Consider $k$ pairs of tuples $(s_1,t_1,d_1),\ldots,(s_k,t_k,d_k)\in V\times V\times\mathbb R_{\ge0}$ indicating that we want to send $d_j$ units of commodity $j\in[k]$ from $s_j$ to $t_j$, and let $\epsilon\in(0,1)$ be an error parameter. If $G$ is a $\phi$-expander, then there is a deterministic algorithm that either certifies that there is no feasible $k$-commodity flow routing the demand, or computes a $k$-commodity flow with congestion $1+\epsilon$. The algorithm runs in time $(m+\epsilon^{-3}k^3D)\cdot\textup{poly}(1/\phi)\cdot2^{O(\sqrt{\log n\log\log n})}$ where $D=\sum_jd_j$.
\end{restatable}

Our results are obtained by localizing Sherman's approximate maximum flow algorithm \cite{Sherman13}.
We remark that Sherman's algorithm and our adaptation are also immediately parallelizable, though a recent parallel local algorithm based on push-relabel has also been discovered~\cite{chen2024parallel}.


\subsection{Our Techniques}

Our starting point is Sherman's initial (or ``unaccelerated'') $\tilde{O}(m/\epsilon^2)$ time algorithm~\cite{Sherman13} with a trivial ``preconditioner'', which we interpret in the Multiplicative Weights Update framework. (We assume no prior background of Sherman's algorithm, and therefore do not discuss the role of preconditioning in Sherman's algorithm.)

At a high level, Sherman's algorithm (with a trivial preconditioner) has a natural interpretation as follows. Each vertex has a \emph{potential} that is initially $0$. On each of $O(\epsilon^{-2}\log n)$ iterations, each edge independently sends \emph{full} flow from the higher-potential endpoint to the lower-potential endpoint (with ties broken arbitrarily). Then, each vertex independently updates its potential based on how much flow it receives; note that potentials can be positive or negative. Intuitively, the more flow a vertex receives, the larger its new potential should be, since it should be more eager to push its flow outward on the next iteration. The actual update to potentials can be viewed as a step in Multiplicative Weights Update (MWU). Finally, at the end of the $O(\epsilon^{-2}\log n)$ iterations, Sherman's algorithm outputs the average of the flows computed on each iteration.

Unfortunately, Sherman's algorithm is not local because on each iteration, \emph{every} edge sends full flow in some direction. One possible modification is to send \emph{no} flow along an edge whose endpoints have equal potential, so in particular, the first flow is the empty flow. This change still guarantees correctness, but vertices can have different potentials very quickly into the algorithm, so the flows in the later iterations may still use all edges. Our main insight is to \emph{round} the potentials of vertices that are small enough in absolute value to $0$. Intuitively, vertices whose potentials are close to $0$ have small flow excess or deficit (cumulatively until the current iteration), and since their potentials are all rounded to the same value, there is no flow sent between such vertices. This strategy is reminiscent of local push-relabel algorithms, where vertices with small enough excess do not push flow outward. Conversely, if a vertex has large enough potential to avoid rounding, then it must have large excess or deficit. We charge the running time of sending flow from this vertex to this excess or deficit in a way that ensures the total charge is bounded.

To show that rounding the potentials still preserves correctness, we prove a variant of MWU that allows for (additively) \emph{approximate} multiplicative weights, which may be of independent interest.

\subsection{Preliminaries}

In this paper, all graphs are unit-capacity and undirected, and the input graph is denoted $G=(V,E)$. We assume that each edge $(u,v)\in E$ has a fixed and arbitrary orientation. The degree $\deg(v)$ of a vertex is the number of incident edges, and for a set of vertices $S\subseteq V$, its volume is defined as $\textbf{\textup{vol}}(S)=\sum_{v\in S}\deg(v)$.

A \emph{flow} is a vector $f\in\mathbb R^E$ that sends, for each (oriented) edge $(u,v)\in E$, $|f(u,v)|$ flow in the direction $(u,v)$ if $f(u,v)>0$ and in the direction $(v,u)$ if $f(u,v)<0$. A flow has \emph{congestion} $\kappa\ge0$ if $|f(u,v)|\le\kappa$ for all $(u,v)\in E$, and it is \emph{feasible} if it has congestion $1$.

A \emph{source function} is a function $b:V\to\mathbb R$, also treated as a vector in $\mathbb R^V$. While source functions in local algorithms are usually non-negative, we work with the more general setting of allowing negative sources. A \emph{demand} is a source function whose entries sum to zero. For a set of vertices $S\subseteq V$, denote $b(S)=\sum_{v\in S}b(v)$. Let $B\in\mathbb R^{V\times E}$ be the incidence matrix of $G$, defined as the matrix where column $(u,v)\in E$ has entry $+1$ at row $u$, entry $-1$ at row $v$, and entry $0$ everywhere else. A flow \emph{routes} demand $b\in\mathbb R^V$ if $Bf=b$. A source function $b$ is \emph{infeasible} if there is no feasible flow $f$ routing demand $b$.

Given a vertex set $S\subseteq V$, let $\partial S\subseteq E$ be the set of edges in $G$ with exactly one endpoint in $S$, and let $\delta S=|\partial S|$.
The following fact, based on flow-cut duality, establishes a \emph{certificate} of infeasibility for a source function $b$.
\begin{fact}
For any source function $b\in\mathbb R^V$, if a vertex set $S\subseteq V$ satisfies $|b(S)|>\delta S$, then $b$ is infeasible.
\end{fact}
\begin{proof}
Suppose for contradiction that there is a feasible flow $f$ routing demand $b$. Observe that $b(S)$ is exactly the net flow leaving the set $S$, which is at most $\delta S$ in absolute value since $f$ is feasible. In other words, $|b(S)|\le\delta S$, a contradiction.
\end{proof}

A \emph{$k$-commodity flow} is a matrix $f\in\mathbb R^{E\times[k]}$ with columns $f_1,\ldots,f_k$, where $f_j\in\mathbb R^E$ is the flow for commodity $j\in[k]$. A $k$-commodity flow has congestion $\kappa\ge0$ if $\sum_{j\in[k]}|f_j(u,v)|\le\kappa$ for all $(u,v)\in E$, and it is \emph{feasible} if it has congestion $1$. A \emph{$k$-commodity source function} is a matrix $b\in\mathbb R^{V\times[k]}$ with columns $b_1,\ldots,b_k\in\mathbb R^V$. A \emph{$k$-commodity demand} is a $k$-commodity source function whose columns $b_1,\ldots,b_k$ are demands. A $k$-commodity flow $f$ \emph{routes} $k$-commodity demand $b$ if $Bf_j=b_j$ for all $j\in[k]$. A $k$-commodity demand $b$ is \emph{infeasible} if there is no feasible $k$-commodity flow $f$ satisfying $Bf_j= b_j$ for all $j\in[k]$. We remark that since the $k$-commodity flow-cut gap is $\Omega(\log k)$, there is no certificate of infeasibility in the form of a cut $\partial S$.

For a vector or matrix $x$, define the $\ell_0$-norm $\|x\|_0$ as the number of nonzero entries in $x$, and the $\ell_1$-norm $\|x\|_1$ as the sum of absolute values of entries in $x$. For a vertex $v\in V$, define $\mathbbm 1_v\in\mathbb R^V$ as the vector with entry $1$ at $v$ and entry $0$ everywhere else.

\section{Local Sherman's Algorithm}

We begin with a quick refresher intended for readers familiar with Sherman's algorithm. Recall that it uses a \emph{congestion approximator}, which is a collection of cuts that approximates the worst cut in any demand. In a $\phi$-expander, the set of \emph{degree} cuts is a $1/\phi$-approximate congestion approximator. Even if the graph is not a $\phi$-expander, we can still use the trivial congestion approximator in Sherman's algorithm. The guarantee we obtain is that each vertex has small residual demand (proportional to its degree), but the residual demand may not be routable with low congestion if the graph is not an expander.

\subsection{Single-commodity}

In this section, we present a single-commodity version as a warm-up. We believe that even the single-commodity case is interesting for our local algorithm, and the goal of this section is to highlight our new ideas in a relatively simple setting.

\begin{theorem}[Local Sherman: Single-commodity]\label{thm:main-single}
There is an algorithm that, given a unit-capacity undirected graph $G=(V,E)$, any source function $b\in\mathbb R^V$, and error parameter $\epsilon\in(0,1)$, outputs either
 \begin{enumerate}
 \item A certificate of infeasibility $S\subseteq V$ with $|b(S)|>\delta S$, or
 \item A flow $\bar f\in\mathbb R^E$ routing demand $b-r$ for some (residual) source function $r\in\mathbb R^V$ satisfying $|r(v)|\le\epsilon\deg(v)$ for all $v\in V$.
 \end{enumerate}
The algorithm runs in time $O(\epsilon^{-2}\log n\cdot(\left\lVert b\right\rVert_0+\epsilon^{-1}\left\lVert b\right\rVert_1))$. Moreover, if an infeasibility certificate $S\subseteq V$ is returned, then $\textbf{\textup{vol}}(S)\le O(\epsilon^{-1}\left\lVert b\right\rVert_1)$.
\end{theorem}

Recall that Sherman's algorithm maintains vertex \emph{potentials} and on each iteration, sends flow across each edge from higher potential to lower potential. Our main idea is to \emph{round} potentials that are small enough (in absolute value) to $0$. More precisely, we write the potentials as a difference of (nonnegative) multiplicative weights, and separately round each multiplicative weight to $0$ if it is small enough.

In the algorithm below, we highlight in \textcolor{blue}{blue} the effects of rounding potentials and how they play into the Multiplicative Weights Update analysis.

\paragraph{Main Algorithm.}
Let $b\in\mathbb R^V$ be the input source function. Initialize weights $w_{v,+}^1=w^1_{v,-}=1$ for each vertex $v\in V$. For round $i=1,2,\ldots,T$ where $T=\Theta(\alpha^{-2}\log n)$ for some $\alpha\le1/4$, we do the following.
 \begin{enumerate}
 \item \textcolor{blue}{Define the \emph{rounded} weights $\tilde w^i_{v,\pm}$ as
 \[
 \tilde w_{v,\pm}^i=\begin{cases}
w_{v,\pm}^i & \text{if }w_{v,\pm}^i\ge n\\
0 & \text{if }w_{v,\pm}^i<n.\\
\end{cases}
 \]
In other words, $\tilde w^i_{v,\pm}$ is obtained by rounding low values to $0$.
} \item Define the \textcolor{blue}{rounded} potential $\tilde\phi^i\in\mathbb R^V$ as
\[ \tilde\phi^i_v=\frac{\tilde w^i_{v,+}-\tilde w^i_{v,-}}{\deg(v)} .\]
 \item Define flow $f^i\in\mathbb R^E$ such that for each edge $(u,v)$, $f^i(u,v)$ flows from higher to lower \textcolor{blue}{rounded} potential at maximum capacity. That is, for every edge $(u,v)\in E$,
 \[
 f^i(u,v)=\begin{cases}
+1 & \text{if }\tilde\phi_u^i > \tilde\phi_v^i\\
-1 & \text{if }\tilde\phi_u^i<\tilde\phi_v^i\\
0 & \text{otherwise.}
\end{cases}
 \]\label{item:step-flow-single}
 \item If $\langle\tilde\phi^i, b\rangle>\langle\tilde\phi^i,Bf^i\rangle$, then terminate and return a certificate of infeasibility $S\subseteq V$ with $|b(S)|>\delta S$, to be described later.\label{item:step-terminate-single}
 \item For each $v\in V$, let
 \[
r_v^i=\frac{ b(v)-(Bf^i)_v}{\deg(v)},
 \]
 be the \emph{relative total excess} at $v$ compared to the degree in round $i$.
 \item Update the (real) weights 
 \[
w^{i+1}_{v,\pm}=w^i_{v,\pm}(1\pm\alpha r_v^i).
 \]\label{item:step-update-single}
 \end{enumerate}

\paragraph{Infeasibility certificate.}

If the termination condition holds, then we show that there is an infeasibility certificate\ $S\subseteq V$ with $|b(S)|>\delta S$, where $S$ has the following form:

\begin{lemma}[Infeasibility certificate]\label{lem:sweep-cut}
Suppose that $\langle\tilde\phi^i, b\rangle>\langle\tilde\phi^i,Bf^i\rangle$ on some iteration $i$.  For real number $x$, define $V_{>x}=\{v\in V:\tilde\phi^i_v>x\}$ and $V_{<x}=\{v\in V:\tilde\phi^i_v>x\}$. There exists an infeasibility certificate $S\subseteq V$ (i.e., $|b(S)|>\delta S$) of the form $S=V_{>x}$ for some $x\ge0$ or $S=V_{<x}$ for some $x\le0$.
\end{lemma}
\begin{proof}
We prove the chain of relations
\begin{equation}
\int_{x=0}^{\infty} b(V_{>x})dx-\int_{x=-\infty}^0b(V_{<x}) dx=\langle\tilde\phi^i, b\rangle>\langle\tilde\phi^i,Bf^i\rangle=\int_{x=0}^{\infty}\delta(V_{>x})dx+\int_{x=0}^{-\infty}\delta(V_{<x})dx.\label{eq:chain}
\end{equation}
We first show how to finish the proof assuming (\ref{eq:chain}): we have either
\[ \int_{x=0}^{\infty} b(V_{>x})dx>\int_{x=0}^{\infty}\delta(V_{>x})dx\qquad\text{or}\qquad-\int_{x=-\infty}^0b(V_{<x})>\int_{x=0}^{-\infty}\delta(V_{<x})dx. \]

In the former case, there must exist $x\ge0$ with $b(V_{>x})>\delta(V_{>x})$, so $V_{>x}$ is our desired set. In the latter case, there must exist $x\le0$ with $-b(V_{<x})>\delta(V_{<x})$, so $V_{<x}$ is our desired set. 

To prove the chain of relations (\ref{eq:chain}), we start with 
\begin{align*}
&\int_{x=0}^{\infty} b(V_{>x})dx-\int_{x=-\infty}^0b(V_{<x}) dx\\
&=\int_{x=0}^{\infty}\left(\sum_{v\in V} b(v)\cdot\mathbbm 1\{\tilde\phi^i_v>x\}\right)dx-\int_{x=-\infty}^0\left(\sum_{v\in V} b(v)\cdot\mathbbm 1\{\tilde\phi^i_v<x\}\right)dx\\
 & =\sum_{v\in V} b(v)\left(\int_{x=0}^{\infty}\mathbbm 1\{\tilde\phi^i_v>x\}dx-\int_{x=-\infty}^0\mathbbm 1\{\tilde\phi^i_v<x\}dx\right)\\
 & =\sum_{v\in V} b(v)\cdot(\max\{0,\tilde\phi^i_v\}-\max\{0,-\tilde\phi^i_v\})\\
 & =\sum_{v\in V} b(v)\cdot\tilde\phi^i_v\\
 & =\langle\tilde\phi^i, b\rangle.
\end{align*}
By definition of the flow $f^i$, 
\begin{align*}
\langle\tilde\phi^i,Bf^i\rangle & =\sum_{(u,v)\in E}\big|\tilde\phi^i_u-\tilde\phi^i_v\big|\\
 & =\sum_{(u,v)\in E}\int_{x=-\infty}^{\infty}\mathbbm 1\{(u,v)\in\partial(V_{>x})\}dx\\
 & =\int_{x=-\infty}^{\infty}\sum_{(u,v)\in E}\mathbbm 1\{(u,v)\in\partial(V_{>x})\}dx\\
 & =\int_{x=-\infty}^{\infty}\delta(V_{>x})dx\\
 & =\int_{x=-\infty}^0\delta(V_{>x})dx+\int_{x=0}^\infty\delta(V_{>x})dx\\
 & =\int_{x=-\infty}^0\delta(V_{<x})dx+\int_{x=0}^\infty\delta(V_{>x})dx.
\end{align*}
Together with the assumption $\langle\phi^i, b\rangle>\langle\phi^i,Bf^i\rangle$,
this establishes (\ref{eq:chain}).
\end{proof}

\paragraph{Multiplicative Weights Update analysis.}
While Sherman's vanilla algorithm can be analyzed by standard Multiplicative Weights Update (MWU), we need a more refined MWU analysis that allows for (additively) \emph{approximate} multiplicative weights when calculating the gain on each iteration.
\begin{restatable}[MWU with Approximate Weights]{theorem}{RoundedMWU}
\label{thm:mwu-onesided} Let $J$
be a set of indices, let $\alpha\le 1/4$ be an error parameter, and let $\nu$ be a value parameter.
Consider the following algorithm: 
\begin{enumerate}
\item Set $w_j^{1}\gets1$ for all $j\in J$ 
\item For $i=1,2,\ldots,T$ where $T=O(\alpha^{-2}\log|J|)$: 
\begin{enumerate}
\item The algorithm is given a ``gain'' vector $g^i\in\mathbb{R}^J$
satisfying $\|g^i\|_{\infty}\le2$ \textcolor{blue}{and an approximate weight vector $\tilde w^i\in\mathbb R^J$ satisfying $\left\lVert w^i-\tilde w^i\right\rVert_\infty\le|J|^{O(1)}$} and $\langle g^i,\tilde w^i\rangle\le0$\label{mwu1-onesided} 
\item For each $j\in J$, set $w_j^{i+1}\gets w_j^i(1+\alpha g_j^i)=\prod_{i'\in[i]}(1+\alpha g^{i'}_j)$
\end{enumerate}
\end{enumerate}
At the end of the algorithm, we have $\frac{1}{T}\sum_{i\in[T]}g_j^i\le5\alpha$
for all $j\in J$.
\end{restatable}
\begin{proof}
Let $W^i=\sum_{j\in J}w^i_j$ be the total weight of all experts at the end of iteration $i$. For any $j\in J$, we have
\begin{align*}
W^{T+1}\ge w^{T+1}_j&=\prod_{i\in[T]}(1+\alpha g^i_j)
\\&\ge\prod_{i\in[T]}\exp(\alpha g^i_j-(\alpha g^i_j)^2) && \text{since }|\alpha g^i_j|\le1/2
\\&=\exp\bigg(\alpha\sum_{i\in[T]}g^i_j-\alpha^2\sum_{i\in[T]}(g^i_j)^2\bigg)
\\&\ge\exp\bigg(\alpha\sum_{i\in[T]}g^i_j-4\alpha^2T\bigg). && \text{since }|g^i_j|\le2
\end{align*}
On the other hand, for all $i\in[T]$, we have
\begin{align*}
W^{i+1}=\sum_{j\in J}w^{i+1}_j&=\sum_{j\in J}w^i_j(1+\alpha g^i_j)
\\&\le\sum_{j\in J}(\tilde w^i_j+\textcolor{blue}{|J|^{O(1)}})(1+\alpha g^i_j) && \textcolor{blue}{\text{since }\left\lVert w^i-\tilde w^i\right\rVert_\infty\le|J|^{O(1)}}
\\&=\sum_{j\in J}\tilde w^i_j(1+\alpha g^i_j)+\textcolor{blue}{|J|^{O(1)}\sum_{j\in J}(1+\alpha g^i_j)}
\\&\le\sum_{j\in J}\tilde w^i_j+\textcolor{blue}{|J|^{O(1)}\sum_{j\in J}(1+\alpha g^i_j)} && \text{since }\langle g^i,\tilde w^i\rangle\le0
\\&\le\sum_{j\in J}w^i_j+\textcolor{blue}{|J|^{O(1)}\cdot 1.5|J|}. && \text{since }|\alpha g^i_j|\le1/2
\end{align*}
Combining the two chains of inequalities, we obtain
\[ \exp\bigg(\alpha\sum_{i\in[T]}g^i_j-4\alpha^2T\bigg)\le W^{T+1}\le\sum_{j\in J}w^1_j+\textcolor{blue}{1.5T|J|^{O(1)}}=|J|+\textcolor{blue}{1.5T|J|^{O(1)}} .\]
It follows that
\[ \alpha\sum_{i\in[T]}g^i_j\le\ln(|J|+\textcolor{blue}{1.5T|J|^{O(1)}})+4\alpha^2T\le5\alpha^2T \]
for large enough $T=O(\alpha^{-2}\log|J|)$, and dividing both sides by $\alpha T$ finishes the proof.
\end{proof}

To apply \Cref{thm:mwu-onesided} into our setting, we define $J=V\times\{+,-\}$.
We use the same weights $w^i_{v,\pm}$ and error parameter $\alpha$ as
the algorithm. For each iteration $i$ and
$v\in V$, we define 
\[
g^i_{v,\pm}=\pm r^i_v=\pm\frac{ b(v)-(Bf^i)_v}{\deg(v)}.
\]
Observe that the weights $w_{v,\pm}^i$ are updated exactly as $w^{i+1}_{v,\pm}\gets w^i_{v,\pm}(1+\alpha g^i_{v,\pm})$.
With this setting, we show that our vectors $g^i,w^i,\tilde w^i$ satisfy
the condition in Step \ref{mwu1-onesided} of \Cref{thm:mwu-onesided}.

For the lemma below, we impose the assumption $|b(v)|\le\deg(v)$ for all $v\in V$, which is easy to check and, if violated, means that the algorithm can simply output $S=\{v\}$ as a certificate of infeasibility.
\begin{lemma}
Assume that the input source function $b$ satisfies $|b(v)|\le\deg(v)$ for all $v\in V$. For each $i$, we have $\|g^i\|_{\infty}\le2$, \textcolor{blue}{$\left\lVert w^i-\tilde w^i\right\rVert_\infty\le n$}, and $\langle g^i,\tilde w^i\rangle\le0$ if the algorithm does not terminate with an infeasibility certificate. 
\end{lemma}

\begin{proof}
To show $\|g^i\|_{\infty}\le2$, we have
\[
|g^i_{v,\pm}|=|r^i_{v,\pm}|=\left|\frac{ b(v)-(Bf^i)_v}{\deg(v)}\right|\le\left|\frac{ b(v)}{\deg(v)}\right|+\left|\frac{f^i(v)}{\deg(v)}\right|\le1+1.
\]
To see why the last inequality holds, we have $|b(v)|\le\deg(v)$
for all $v\in V$ by assumption, and $|f^i(v)|\le\deg(v)$
because each $f^i$ is feasible.

\textcolor{blue}{To show $\left\lVert w^i-\tilde w^i\right\rVert_\infty\le n$, observe that for each $v\in V$, we have $|w^i_{v,\pm}-\tilde w^i_{v,\pm}|\le n$ from the construction of $\tilde w^i$.}

To show $\langle g^i,\tilde w^i\rangle\le0$, observe that for each $v\in V$,
\[ g^i_{v,+}\cdot \tilde w^i_{v,+}+g^i_{v,-}\cdot \tilde w^i_{v,-}=\frac{ b(v)-(Bf^i)_v}{\deg(v)}\cdot(\tilde w^i_{v,+}-\tilde w^i_{v,-})=(b(v)-(Bf^i)_v)\cdot\tilde\phi_v ,\]
so $\langle g^i,\tilde w^i\rangle=\langle\tilde \phi^i, b-Bf^i\rangle$. Since the algorithm does not terminate with an infeasibility certificate, we must have $\langle\tilde \phi^i, b\rangle\le\langle\tilde \phi^i,Bf^i\rangle$, so $\langle g^i,\tilde w^i\rangle\le0$.
\end{proof}
From the above, we have verified that our algorithm is indeed captured
by the MWU algorithm. Next, define the average flow $\bar{f}=\frac{1}{T}\sum_{i=1}^Tf^i\in\mathbb{R}^{E}$. We now prove the two-sided error guarantee $|b(v)-(B\bar f)_v|\le5\alpha\deg(v)$.
\begin{lemma}[Approximation]
The flow $\bar{f}$ satisfies $|b(v)-(B\bar f)_v|\le5\alpha\deg(v)$ for all $v\in V$. In other words, $\bar f$ satisfies input source function $b$ up to two-sided error $5\alpha\deg(v)$.
\end{lemma}

\begin{proof}
Define $\bar{r}=\frac{1}{T}\sum_{i=1}^Tr^i\in\mathbb{R}^V$.
First, we prove that $|\bar{r}_v|\le5\alpha$ for all $v\in V$.
This is because 
\[
\pm\bar{r}_v=\frac{1}{T}\sum_{i\in[T]}\pm r_v^i=\frac{1}{T}\sum_{i\in[T]}g^i_{v,\pm}\le5\alpha,
\]
where the last inequality is precisely the guarantee of the MWU algorithm
from \Cref{thm:mwu-onesided}. By definition of $\bar r$,
\[ \bar r_v=\frac1T\sum_{i\in[T]}r_v^i=\frac1T\sum_{i\in[T]}\frac{ b(v)-(Bf^i)_v}{\deg(v)}=\frac{ b(v)-(B\bar f)_v}{\deg(v)} ,\]
and together with $|\bar{r}_v|\le5\alpha$ concludes the proof.
\end{proof}

\paragraph{Running time bound.}

Finally, we analyze the running time, which we show only depends on $\left\lVert b\right\rVert_0$ and $\left\lVert b\right\rVert_1$. We remark that everything so far generalizes to capacitated graphs, and the running time bound is the only place that requires a unit-capacity graph.

To avoid clutter, define $f^{\le i}=f^1+\cdots+f^i$ and  $r^{\le i}=r^1+\cdots+r^i$.
\begin{lemma}\label{lem:phi-helper}
If $\tilde\phi^i_v>0$, then $r^{\le i-1}_v\ge\frac{\ln n}\alpha$. Similarly, if $\tilde\phi^i_v<0$, then $r^{\le i-1}_v\le-\frac{\ln n}\alpha$.
\end{lemma}
\begin{proof}
If $\tilde\phi^i_v>0$, then since $\tilde\phi^i_v=\deg(v)\cdot(\tilde w^i_{v,+}-\tilde w^i_{v,-})>0$, we have $\tilde w^i_{v,+}>0$, so $w^i_{v,+}\ge n$. Similarly, if $\tilde\phi^i_v<0$, then since $\tilde\phi^i_v=\deg(v)\cdot(\tilde w^i_{v,+}-\tilde w^i_{v,-})<0$, we have $\tilde w^i_{v,-}>0$, so $w^i_{v,-}\ge n$. In both cases, we have
\[ n\le w^i_{v,\pm}=\prod_{i'\le i-1}(1\pm\alpha r^{i'}_v)\le\prod_{i'\le i-1}\exp(\pm\alpha r^{i'}_v)=\exp(\pm\alpha r^{\le i-1}_v) ,\]
and taking the logarithm and diving by $\alpha$ gives $\pm r^{\le i-1}_v\ge\frac{\ln n}\alpha$, establishing both cases.
\end{proof}
\begin{lemma}\label{lem:r-helper}
For each iteration $i$, we have $\sum_{v\in V}|r^{\le i}_v|\cdot\deg(v)\le i\cdot\left\lVert b\right\rVert_1$.
\end{lemma}
\begin{proof}
We prove the statement by induction on $i\ge0$, with the trivial base case $i=0$. Suppose by induction that $\sum_{v\in V}|r^{\le i-1}_v|\cdot\deg(v)\le (i-1)\cdot\left\lVert b\right\rVert_1$. For each vertex $v\in V$, we have
\[ r^{\le i}_v\cdot\deg(v)=r^{\le i-1}_v\cdot\deg(v)+r^i_v\cdot\deg(v)=r^{\le i-1}_v\cdot\deg(v)+b(v)-(Bf^i)_v. \]
Therefore,
\[ \sum_{v\in V}|r^{\le i}_v|\cdot\deg(v)=\sum_{v\in V}|r^{\le i-1}_v\cdot\deg(v)+b(v)-(Bf^i)_v|\le\sum_{v\in V}|r^{\le i-1}_v\cdot\deg(v)-(Bf^i)_v|+\left\lVert b\right\rVert_1 .\]
Our goal is to show that
\[ \sum_{v\in V}|r^{\le i-1}_v\cdot\deg(v)-(Bf^i)_v|\le\sum_{v\in V}|r^{\le i-1}_v\cdot\deg(v)| ,\]
which together with the inductive assumption completes the induction. To establish the above inequality, imagine carrying out the flow $f^i$ edge-by-edge. More formally, enumerate the edges of $E$ in some arbitrary order $e_1,\ldots,e_m$, and let $E_{\le \ell}=\{e_1,\ldots,e_\ell\}$ be the first $\ell$ edges. Let $f^i|_{E_{\le\ell}}$ be the flow $f^i$ restricted to edges $E_{\le\ell}$, i.e., zero out all coordinates outside $E_{\le\ell}$. We view $f^i|_{E_{\le \ell}}$ as the flow $f^i$ with the first $\ell$ edges carried out. Since $f^i|_{E_{\le0}}=0$ and $f^i|_{E_{\le m}}=f^i$, it suffices to show that the value
\[ \Phi(\ell)=\sum_{v\in V}|r^{\le i-1}_v\cdot\deg(v)-(Bf^i|_{E_{\le\ell}})_v| \]
is monotonically decreasing as $\ell$ increases from $0$ to $m$. For a given $\ell\in[m]$, write $e_\ell=(u,v)$. If $f^i(u,v)=0$, then $f^i|_{E_{\le\ell-1}}=f^i|_{E_{\le\ell}}$ and $\Phi(\ell-1)=\Phi(\ell)$, as promised. Otherwise, suppose without loss of generality that $f^i(u,v)=1$, which means $\tilde\phi^i_u>\tilde\phi^i_v$ and $Bf^i|_{E_{\le\ell}}=Bf^i|_{E_{\le\ell-1}}+\mathbbm 1_u-\mathbbm 1_v$. In particular,
\begin{align}
\Phi(\ell)-\Phi(\ell-1)=\quad&|r^{\le i-1}_u\cdot\deg(u)-(Bf^i|_{E_{\le\ell}})_u|-|r^{\le i-1}_u\cdot\deg(u)-(Bf^i|_{E_{\le\ell-1}})_u|\nonumber
\\+\,&|r^{\le i-1}_v\cdot\deg(v)-(Bf^i|_{E_{\le\ell}})_v|-|r^{\le i-1}_v\cdot\deg(v)-(Bf^i|_{E_{\le\ell-1}})_v|\nonumber
\\=\quad&|r^{\le i-1}_u\cdot\deg(u)-(Bf^i|_{E_{\le\ell-1}})_u-1|-|r^{\le i-1}_u\cdot\deg(u)-(Bf^i|_{E_{\le\ell-1}})_u|\label{eq:phi-1}
\\+\,&|r^{\le i-1}_v\cdot\deg(v)-(Bf^i|_{E_{\le\ell-1}})_v+1|-|r^{\le i-1}_v\cdot\deg(v)-(Bf^i|_{E_{\le\ell-1}})_v|.\label{eq:phi-2}
\end{align}
Our goal is to show that either
\[ r^{\le i-1}_u\cdot\deg(u)-(Bf^i|_{E_{\le\ell-1}})_u\ge1\qquad\text{or}\qquad r^{\le i-1}_v\cdot\deg(v)-(Bf^i|_{E_{\le\ell-1}})_v\le-1 ,\]
which implies that the expression in either (\ref{eq:phi-1}) or (\ref{eq:phi-2}) equals $-1$, and since the other expression is at most $+1$, we obtain our desired $\Phi(\ell)-\Phi(\ell-1)\le0$.

Since $\tilde\phi^i_u>\tilde\phi^i_v$, we have either $\tilde\phi^i_u>0$ or $\tilde\phi^i_v<0$. If $\tilde\phi^i_u>0$, then by \Cref{lem:phi-helper}, we have
\[ r^{\le i-1}_u\ge\frac{\ln n}\alpha\ge\frac1\alpha\ge2 .\]
Since $f^i|_{E_{\le\ell-1}}$ is a feasible flow, we have $(Bf^i|_{E_{\le\ell-1}})_u\le\deg(u)$, so
\[ r^{\le i-1}_u\cdot\deg(u)-(Bf^i|_{E_{\le\ell-1}})_u\ge2\deg(u)-\deg(u)=\deg(u)\ge1 ,\]
as promised. Similarly, if $\tilde\phi^i_v<0$, then by \Cref{lem:phi-helper}, we have
\[ r^{\le i-1}_v\le-\frac{\ln n}\alpha\le-\frac1\alpha\le-2. \]
Since $f^i|_{E_{\le\ell-1}}$ is a feasible flow, we have $(Bf^i|_{E_{\le\ell-1}})_v\ge-\deg(v)$, so
\[ r^{\le i-1}_v\cdot\deg(v)-(Bf^i|_{E_{\le\ell-1}})_v\le-2\deg(v)+\deg(v)=-\deg(v)\le-1, \]
 as promised. This completes the induction and hence the proof.
\end{proof}

\begin{lemma}\label{lem:time-helper}
On a unit-capacity graph, Sherman's algorithm can be implemented in $O(T\left\lVert b\right\rVert_0+T^2\left\lVert b\right\rVert_1\alpha/\ln n)$ time. Moreover, if an infeasibility certificate $S\subseteq V$ is returned, then $\textbf{\textup{vol}}(S)\le T\left\lVert b\right\rVert_1\alpha/\ln n$.
\end{lemma}
\begin{proof}
Fix an iteration $i\in[T]$, and define $A^i=\{v\in V:\tilde\phi_v^i\ne0\}$. By \Cref{lem:phi-helper}, we have $|r^{\le i-1}_v|\ge\frac{\ln n}\alpha$ for all $v\in A^i$. By \Cref{lem:r-helper}, we have $\sum_{v\in V}|r^{\le i-1}_v|\cdot\deg(v)\le(i-1)\cdot\left\lVert b\right\rVert_1$. It follows that
\[ (i-1)\cdot\left\lVert b\right\rVert_1\ge\sum_{v\in V}|r^{\le i-1}_v|\cdot\deg(v)\ge\sum_{v\in A^i}|r^{\le i-1}_v|\cdot\deg(v)\ge\sum_{v\in A^i}\frac{\ln n}\alpha\cdot\deg(v)=\frac{\ln n}\alpha\cdot\textbf{\textup{vol}}(A^i) ,\]
so $\textbf{\textup{vol}}(A^i)\le T\left\lVert b\right\rVert_1\alpha/\ln n$.

Each edge with nonzero flow in $f^i$ is incident to a vertex in $A^i$, so step~(\ref{item:step-flow-single}) of iteration $i\in[T]$ can be implemented in $O(\textbf{\textup{vol}}(A^i))$ time on a unit-capacity graph. For step~(\ref{item:step-update-single}), the only vertices with nonzero $r^i_v$ are either those with nonzero $b(v)$ or those incident to an edge with nonzero flow. The number of these vertices can be bounded by $\left\lVert b\right\rVert_0$ and $\textbf{\textup{vol}}(A^i)$, respectively. It follows that the running time of iteration $t$ is $O(\left\lVert b\right\rVert_0+\textbf{\textup{vol}}(A^i))$ if step~(\ref{item:step-terminate-single}) does not hold. If step~(\ref{item:step-terminate-single}) holds, then by \Cref{lem:sweep-cut}, the certificate of infeasibility $S\subseteq V$ satisfies either $S=V_{>x}\subseteq A^i$ for some $x\ge0$, or $S=V_{<x}\subseteq A^i$ for some $x\le0$, so $\textbf{\textup{vol}}(S)\le\textbf{\textup{vol}}(A^i)$. To find the desired set $V_{>x}$ or $V_{<x}$, sort the vertex set $V_{>0}$ by decreasing $\tilde\phi^i_v$ and the vertex set $V_{<0}$ by increasing $\tilde\phi^i_v$, which takes $O(|A^i|\log|A^i|)$ time. Then, compute the values of $b(S)$ and $\delta(S)$ for all prefixes $S$ of the two sorted lists, which takes $O(\textbf{\textup{vol}}(A^i))$ time by sweeping through the lists.

Therefore, over all iterations $i\in[T]$, the total running time is $O(T\left\lVert b\right\rVert_0+\sum_{i\in[T]}\textbf{\textup{vol}}(A^i))$ outside of computing the certificate of infeasibility, and $O(|A^{i^*}|\log|A^{i^*}|+\textbf{\textup{vol}}(A^{i^*}))\le O(\textbf{\textup{vol}}(A^{i^*})\cdot T)$ for possibly computing it on some iteration $i^*\in[T]$. Since $\textbf{\textup{vol}}(A^i)\le T\left\lVert b\right\rVert_1\alpha/\ln n$ for all $i\in[T]$, the total running time is $O(T\left\lVert b\right\rVert_0+T^2\left\lVert b\right\rVert_1\alpha/\ln n)$. If a certificate of infeasibility $S\subseteq V$ is returned, then $\textbf{\textup{vol}}(S)\le\textbf{\textup{vol}}(A^{i^*})\le T\left\lVert b\right\rVert_1\alpha/\ln n$.
\end{proof}

\paragraph{Final parameters.}
For error $|b(v)-(B\bar f)_v|\le\epsilon\deg(v)$ on each vertex $v\in V$, we set $\alpha=\epsilon/5\le1/4$ and $T=O(\alpha^{-2}\log n)=O(\epsilon^{-2}\log n)$. The total running time is $O(T\left\lVert b_0\right\rVert+T^2\left\lVert b\right\rVert_1\alpha/\ln n)=O(\epsilon^{-2}\log n\cdot(\left\lVert b\right\rVert_0+\epsilon^{-1}\left\lVert b\right\rVert_1))$. Moreover, if an infeasibility certificate $S\subseteq V$ is returned, then $\textbf{\textup{vol}}(S)\le T\left\lVert b\right\rVert_1\alpha/\ln n=O(\epsilon^{-1}\left\lVert b\right\rVert_1)$.

\subsection{Multi-commodity}
In this section, we present the multi-commodity flow algorithm, establishing our main theorem below.

\begin{theorem}[Local Sherman: Multi-commodity]\label{thm:main-multi}
There is an algorithm that, given a unit-capacity undirected graph $G=(V,E)$, a parameter $k$, any $k$-commodity source function $b\in\mathbb R^{V\times[k]}$, and error parameter $\epsilon\in(0,1)$, outputs either
 \begin{enumerate}
 \item A certificate of infeasibility of the multi-commodity flow, or
 \item A $k$-commodity flow $\bar f\in\mathbb R^{E\times[k]}$ where $\bar f_1,\ldots,\bar f_k\in\mathbb R^E$ route demands $b_1-r_1,\ldots,b_k-r_k\in\mathbb R^V$ for (residual) source functions $r_1,\ldots,r_k\in\mathbb R^V$ satisfying $ |r_j(v)|\le\epsilon\deg(v)$ for all $v\in V,\,j\in[k]$.
 \end{enumerate}
The algorithm runs in time $O(\epsilon^{-2}\log n \cdot(\left\lVert b\right\rVert_0+\epsilon^{-1}\left\lVert b\right\rVert_1))$.
\end{theorem}

Let $b\in\mathbb R^{V\times[k]}$ be the input $k$-commodity source function, and let $\mathcal F$ denote all feasible $k$-commodity flows, i.e., matrices $f\in\mathbb R^{E\times[k]}$ satisfying $\sum_{j\in[k]}|f_j(u,v)|\le 1$ for all edges $(u,v)\in E$. 

The $k$-commodity version introduces another layer of complexity: since the $k$-commodity flow-cut gap is $\Omega(\log k)$, the algorithm can no longer output a single cut to certify infeasibility. To obtain the new certificate of infeasibility, we first investigate the linear program that we wish to solve:
\begin{align}
\text{find }&f\in\mathcal F\nonumber
\\\text{s.t. }& \pm\frac{b_j(v)-(Bf_j)_v}{\deg(v)}\le0 && \forall v\in V,\,j\in[k].\label{eq:LP-constraints}
\end{align}
Note that there is no need to divide $b_j(v)-(Bf_j)_v$ by $\deg(v)$ in (\ref{eq:LP-constraints}), 
but doing so normalizes the inequalities for our purposes.

In the MWU framework, we weight each constraint $(v,j,\pm)$ in (\ref{eq:LP-constraints}) by $w^i_{v,j,\pm}\ge0$ 
and solve the weighted average inequality subject to the ``easy'' constraint $f\in\mathcal F$. If the average constraint is infeasible, then the weights $\phi^i_{v,j,\pm}$ certify infeasibility of the LP.

\paragraph{Main Algorithm.}
Initialize weights $w_{v,j,+}^1=w^1_{v,j,-}=1$ for each $v\in V$ and $j\in[k]$. For each round $i=1,2,\ldots,T$ where $T=\Theta(\alpha^{-2}\log n)$ for some $\alpha\le1/4$, we do the following.
 \begin{enumerate}
 \item Define the \emph{rounded} weights $\tilde w^i_{v,j,\pm}$ as
 \[
 \tilde w_{v,j,\pm}^i=\begin{cases}
w_{v,j,\pm}^i & \text{if }w_{v,j,\pm}^i\ge n\\
0 & \text{if }w_{v,j,\pm}^i<n.\\
\end{cases}
 \]
In other words, $\tilde w^i_{v,j,\pm}$ is obtained by rounding low values to $0$. 
 \item Define the rounded potential $\tilde\phi^i\in\mathbb R^{V\times[k]}$ where \label{item:step-define-multi}
\[ \tilde\phi_{v,j}^i=\frac{\tilde w_{v,j,+}^i-\tilde w^i_{v,j,-}}{\deg(v)} .\]
 \item Define (multi-commodity) flow $f^i\in\mathbb R^{E\times[k]}$ such that for each edge $(u,v)$, $f^i(u,v)$ flows from higher to lower \textcolor{blue}{rounded} potential at maximum capacity \textcolor{blue}{in the commodity $j\in[k]$ maximizing the potential difference}. That is, for every edge $(u,v)\in E$, let $j^*\in[k]$ be the commodity maximizing $|\tilde\phi_{u,j}-\tilde\phi_{v,j}|$ over all $j\in[k]$, and define
 \[
 f^i_{j^*}(u,v)=\begin{cases}
+1 & \text{if }\tilde\phi_{u,j^*}^i>\tilde\phi_{v,j^*}^i\\
-1 & \text{if }\tilde\phi_{u,j^*}^i<\tilde\phi_{v,j^*}^i\\
0 & \text{otherwise.}
\end{cases}
 \]
Define $f^i_j(u,v)=0$ for all $j\ne j^*$. \label{item:step-flow-multi}
 \item If $\langle\tilde\phi^i,\bigoplus_{j\in[k]} b_j\rangle>\langle\tilde\phi^i,\bigoplus_{j\in[k]}Bf^i_j\rangle$, then terminate with $\tilde\phi^i$ as the infeasibility certificate. Here, 
the operator $\bigoplus$ concatenates vectors in $\mathbb R^V$ into a single vector in $\mathbb R^{V\times[k]}$.
 \item For each $v\in V$ and $j\in[k]$, let
 \[
r_{v,j}^i=\frac{b_j(v)-(Bf^i_j)_v}{\deg(v)}
 \]
 be the \emph{relative total excess} in commodity $j$ at $v$ compared to the degree in round $i$.
 \item Update the (real) weights 
 \[
w_{v,j,\pm}^{i+1}=w_{v,j,\pm}^i(1\pm\alpha r_{v,j}^i).
 \] \label{item:step-update-multi}
 \end{enumerate}

\paragraph{Infeasibility certificate.}
If the termination condition holds, then we show that there is an infeasibility certificate.

\begin{lemma}
If $\langle\tilde\phi^i,\bigoplus_{j\in[k]} b_j\rangle>\langle\tilde\phi^i,\bigoplus_{j\in[k]}Bf^i_j\rangle$ on some iteration $i$, then $b$ is infeasible.
\end{lemma}
\begin{proof}
We show that taking a weighted sum of the constraints (\ref{eq:LP-constraints}), where each constraint $(v,j,\pm)$ is weighted by $\tilde w^i_{v,j,\pm}\ge0$, certifies infeasibility of the LP. The left hand side of this weighted sum is
\begin{align}
&\sum_{v\in V}\sum_{j\in[k]}\left(\tilde w^i_{v,j,+}\cdot\frac{b_j(v)-(Bf^i_j)_v}{\deg(v)}-\tilde w^i_{v,j,-}\cdot\frac{b_j(v)-(Bf^i_j)_v}{\deg(v)}\right) \nonumber
\\&=\sum_{v\in V}\sum_{j\in[k]}\left(\tilde\phi^i_{v,j}\cdot(b_j(v)-(Bf^i_j)_v)\right) \nonumber
\\&=\sum_{v\in V}\sum_{j\in[k]}\tilde\phi^i_{v,j}b_j(v)-\sum_{v\in V}\sum_{j\in[k]}\tilde\phi^i_{v,j}(Bf^i_j)_v,\label{eq:average-constraint}
\end{align}
which equals $\langle\tilde\phi^i,\bigoplus_{j\in[k]} b_j\rangle-\langle\tilde\phi^i,\bigoplus_{j\in[k]}Bf^i_j\rangle$ and is positive by assumption.
The first double-summation in (\ref{eq:average-constraint}) is a constant independent of $f\in\mathcal F$, and the second double-summation equals 
\begin{gather}
\sum_{(u,v)\in E}\sum_{j\in[k]}f^i_j(u,v)\cdot (\tilde\phi^i_{u,j}-\tilde\phi^i_{v,j}) .\label{eq:average-constraint-2}
\end{gather}
Observe that the $k$-commodity flow $f^i\in\mathbb R^{E\times[k]}$ is chosen to maximize (\ref{eq:average-constraint-2}) subject to the feasibility constraint $f\in\mathcal F$. It follows that $f^i$ minimizes (\ref{eq:average-constraint}) subject to $f\in\mathcal F$, and yet the value of (\ref{eq:average-constraint}) is still positive. Therefore, the value must be positive for all $f\in\mathcal F$, and since it is a weighted sum of the constraints (\ref{eq:LP-constraints}), we conclude that the LP is infeasible. In particular, there is no flow $f\in\mathcal F$ satisfying $Bf_j= b_j$ for all $j\in[k]$, so $b$ is infeasible. 
\end{proof}
\paragraph{Multiplicative Weights Update analysis.}
Recall the refined MWU statement that allows for approximate multiplicative weights, restated below for convenience.
\RoundedMWU*

To apply \Cref{thm:mwu-onesided} into our setting, we define $J=V\times[k]\times\{+,-\}$.
We use the same weights $w^i_{v,j,\pm}$ and error parameter $\alpha$ as
the algorithm. For each iteration $i$ and
$v\in V$, we define 
\[
g_{v,j,\pm}^i=\pm r_{v,j}^i=\pm\frac{ b_j(v)-(Bf^i_j)_v}{\deg(v)}.
\]
Observe that the weights $w_{v,j,\pm}^i$ are updated exactly as $w_{v,j,\pm}^{i+1}\gets w_{v,j,\pm}^i(1\pm\alpha g_{v,j,\pm}^i)$. 
With this setting, we show that our vectors $g^i,w^i,\tilde w^i$ satisfy
the condition in Step \ref{mwu1-onesided} of \Cref{thm:mwu-onesided}.

For the lemma below, we impose the assumption $|b_j(v)|\le\deg(v)$ for all $v\in V$ and $j\in[k]$, which is easy to check and, if violated, is a quick certificate of infeasibility.
\begin{lemma}
Assume that the input source functions $b_j$ satisfy $b_j(v)\le\deg(v)$ for all $v\in V,\,j\in[k]$. For each $i$, we have $\|g^i\|_{\infty}\le2$, $\left\lVert w^i-\tilde w^i\right\rVert_\infty\le n$, and $\langle g^i,\tilde w^i\rangle\le0$ if the algorithm does not terminate with an infeasibility certificate.
\end{lemma}

\begin{proof}
To show $\|g^i\|_{\infty}\le2$, we have
\[
|g_{v,j,\pm}^i|=|r_{v,j,\pm}^i|=\left|\frac{ b_j(v)-(Bf^i_j)_v}{\deg(v)}\right|\le\left|\frac{ b_j(v)}{\deg(v)}\right|+\left|\frac{f^i_j(v)}{\deg(v)}\right|\le1+1.
\]
To see why the last inequality holds, we have $|b_j(v)|\le\deg(v)$
for all $v\in V$ by assumption, and $|f^i_j(v)|\le\deg(v)$
because each $f^i$ is feasible.

To show $\left\lVert w^i-\tilde w^i\right\rVert_\infty\le n$, observe that for each $v\in V$ and $j\in[k]$, we have $|w^i_{v,j,\pm}-\tilde w^i_{v,j,\pm}|\le n$ from the construction of $\tilde w^i$.

To show $\langle g^i,\tilde w^i\rangle\le0$, observe that for each $v\in V$,
\begin{align*}
\sum_{j\in[k]}(g^i_{v,j,+}\cdot \tilde w_{v,j,+}^i+g^i_{v,j,-}\cdot \tilde w_{v,j,-}^i)&=\sum_{j\in[k]}\frac{ b_j(v)-(Bf^i_j)_v}{\deg(v)}\cdot(\tilde w^i_{v,j,+}-\tilde w^i_{v,j,-})
=\sum_{j\in[k]}(b_j(v)-(Bf^i_j)_v)\cdot\tilde\phi^i_{v,j} ,
\end{align*}
so summed over all $v\in V$, we have $\langle g^i,\tilde w^i\rangle=\langle\tilde\phi^i,\bigoplus_{j\in[k]}(b_j-Bf^i_j)\rangle$. Since the algorithm does not terminate with an infeasibility certificate, we must have $\langle\tilde\phi^i,\bigoplus_{j\in[k]}b_j\rangle\le\langle\tilde\phi^i,\bigoplus_{j\in[k]}Bf^i_j\rangle$, so $\langle g^i,\tilde w^i\rangle\le0$.
\end{proof}

From the above, we have verified that our algorithm is indeed captured
by the MWU algorithm. Next, define the average flows $\bar{f}_j=\frac{1}{T}\sum_{i=1}^Tf^i_j\in\mathbb{R}^{E}$ for all $j\in[k]$. We now prove the two-sided error guarantee $|b_j(v)-(B\bar f_j)_v|\le5\alpha\deg(v)$.
\begin{lemma}[Approximation]
The flow $\bar{f}_j$ satisfies
$|b_j(v)-(B\bar f_j)_v|\le5\alpha\deg(v)$ for all $v\in V,\,j\in[k]$.
In other words, $\bar f_j$ satisfies input source function $b_j$ up to two-sided error $5\alpha\deg(v)$.
\end{lemma}

\begin{proof}
Define $\bar{r}=\frac{1}{T}\sum_{i=1}^Tr^i\in\mathbb{R}^{V\times[k]}$.
First, we prove that $|\bar{r}_{v,j}|\le5\alpha$ for all $v\in V,\,j\in[k]$.
This is because 
\[
\pm\bar{r}_{v,j}=\frac{1}{T}\sum_{i\in[T]}\pm r_{v,j}^i=\frac{1}{T}\sum_{i\in[T]}g_{v,j,\pm}^i\le5\alpha ,
\]
where the last inequality is precisely the guarantee of the MWU algorithm
from \Cref{thm:mwu-onesided}. By definition of $\bar r$,
\[ \bar r_{v,j}=\frac1T\sum_{i\in[T]}r_{v,j}^i=\frac1T\sum_{i\in[T]}\frac{ b_j(v)-(Bf^i_j)_v}{\deg(v)}=\frac{ b_j(v)-(B\bar f_j)_v}{\deg(v)},\]
and together with $|\bar{r}_{v,j}|\le5\alpha$
concludes the proof.
\end{proof}

\paragraph{Running time bound.}

Finally, we analyze the running time, which we show only depends on $\left\lVert b\right\rVert_0=\sum_{j\in[k]}\left\lVert b_j\right\rVert_0$ and $\left\lVert b\right\rVert_1=\sum_{j\in[k]}\left\lVert b_j\right\rVert_1$. We remark that everything so far generalizes to capacitated graphs, and the running time bound is the only place that requires a unit-capacity graph.

To avoid clutter, define $f^{\le i}=f^1+\cdots+f^i$ and  $r^{\le i}=r^1+\cdots+r^i$. The following two helper lemmas are proved identically to \Cref{lem:r-helper,lem:phi-helper}, adding a subscript $j\in[k]$ wherever necessary. Note that unlike the flow $f^i$ in the single-commodity case, the flow $f^i_j$ may send zero flow across edges $(u,v)\in E$ with $\tilde\phi_{u,j}\ne\tilde\phi_{v,j}$, but this does not affect the proof of \Cref{lem:r-helper} since $f^i_j(u,v)=0$ is an easy case.
\begin{lemma}\label{lem:phi-helper-multi}
If $\tilde\phi^i_{v,j}>0$, then $r^{\le i-1}_{v,j}\ge\frac{\ln n}\alpha$. Similarly, if $\tilde\phi^i_{v,j}<0$, then $r^{\le i-1}_{v,j}\le-\frac{\ln n}\alpha$.
\end{lemma}
\begin{lemma}\label{lem:r-helper-multi}
For each iteration $i$, we have $\sum_{v\in V}|r^{\le i}_{v,j}|\cdot\deg(v)\le i\cdot\left\lVert b_j\right\rVert_1$.
\end{lemma}

\begin{lemma}
On a unit-capacity graph, Sherman's algorithm can be implemented in $O(T\left\lVert b\right\rVert_0+T^2\left\lVert b\right\rVert_1\alpha/\ln n)$ time.
\end{lemma}
\begin{proof}
Fix an iteration $i\in[T]$, and define $A^i_j=\{v\in V:\tilde\phi_{v,j}^i\ne0\}$ for all $j\in[k]$. By \Cref{lem:phi-helper-multi}, we have $|r^{\le i-1}_{v,j}|\ge\frac{\ln n}\alpha$ for all $v\in A^i_j$. By \Cref{lem:r-helper-multi}, we have $\sum_{v\in V}|r^{\le i-1}_{v,j}|\cdot\deg(v)\le(i-1)\cdot\left\lVert b_j\right\rVert_1$. It follows that
\[ (i-1)\cdot\left\lVert b_j\right\rVert_1\ge\sum_{v\in V}|r^{\le i-1}_{v,j}|\cdot\deg(v)\ge\sum_{v\in A^i_j}|r^{\le i-1}_{v,j}|\cdot\deg(v)\ge\sum_{v\in A^i_j}\frac{\ln n}\alpha\cdot\deg(v)=\frac{\ln n}\alpha\cdot\textbf{\textup{vol}}(A^i_j) ,\]
so $\textbf{\textup{vol}}(A^i_j)\le T\left\lVert b_j\right\rVert_1\alpha/\ln n$.

Each edge with nonzero flow in $f^i_j$ is incident to a vertex in $A^i_j$, so step~(\ref{item:step-flow-single}) of iteration $i\in[T]$ can be implemented in $O(\sum_{j\in[k]}\textbf{\textup{vol}}(A^i_j))$ time on a unit-capacity graph. For step~(\ref{item:step-update-single}), the only vertices with nonzero $r^i_{v,j}$ are either those with nonzero $b_j(v)$ or those incident to an edge with nonzero flow. The number of these vertices can be bounded by $\left\lVert b\right\rVert_0$ and $\sum_{j\in[k]}\textbf{\textup{vol}}(A^i_j)$, respectively. It follows that the running time of iteration $t$ is $O(\left\lVert b\right\rVert_0+\sum_{j\in[k]}\textbf{\textup{vol}}(A^i_j))$.

Therefore, over all iterations $i\in[T]$, the total running time is $O(T\left\lVert b\right\rVert_0+\sum_{i\in[T]}\sum_{j\in[k]}\textbf{\textup{vol}}(A^i_j))$. Since $\textbf{\textup{vol}}(A^i_j)\le T\left\lVert b_j\right\rVert_1\alpha/\ln n$ for all $i\in[T]$ and $j\in[k]$, the total running time is $O(T\left\lVert b\right\rVert_0+T^2\left\lVert b\right\rVert_1\alpha/\ln n)$.
\end{proof}
\paragraph{Final parameters.}
For error $|b_j(v)-(B\bar f_j)_v|\le\epsilon\deg(v)$ for all $v\in V,\,j\in[k]$, we set $\alpha=\epsilon/5\le1/4$ and $T=O(\alpha^{-2}\log n)=O(\epsilon^{-2}\log n)$. The total running time is $O(T\left\lVert b_0\right\rVert+T^2\left\lVert b\right\rVert_1\alpha/\ln n)=O(\epsilon^{-2}\log n\cdot(\left\lVert b\right\rVert_0+\epsilon^{-1}\left\lVert b\right\rVert_1))$.

\section{Multi-commodity Flow on an Expander}

A \emph{$\phi$-expander} is a graph satisfying $\delta S\ge\phi\min\{\textbf{\textup{vol}}(S),\textbf{\textup{vol}}(V\setminus S)\}$ for all $S\subseteq V$.  In this section, we prove \Cref{thm:expander-main}, restated below.
\Expander*

In a $\phi$-expander, the residual demands from \Cref{thm:main-multi} can be routed to completion, incurring only $\epsilon$ additional congestion as long as the residual demands are small enough. The following is implicit in \cite{chang2024deterministic}, who describe the algorithm as a distributed routing protocol.

\begin{theorem}[Expander Routing]
\label{thm:expander-routing} Consider the $\ell$-commodity demand defined by $\ell$ weighted pairs $(s_{1},t_{1},d_{1}),\ldots,(s_{\ell},t_{\ell},d_{\ell})\in V\times V\times\mathbb{R}_{\ge0}$ indicating that we want to send $d_{j}$ units of commodity $j\in[\ell]$ from $s_{j}$ to $t_{j}$. Assume that for each vertex $v\in V$, the sum of $d_{j}$ over all $j\in[k]$ with $v\in\{s_{j},t_{j}\}$ is at most $\alpha\deg(v)$ for some parameter $\alpha>0$. If $G$ is a $\phi$-expander, then there is a deterministic $(m+\ell)\cdot\textup{poly}(1/\phi)\cdot2^{O(\sqrt{\log n\log\log n})}$ time algorithm that computes an $\ell$-commodity flow routing this demand with congestion $\alpha\cdot\textup{poly}(1/\phi)\cdot2^{O(\sqrt{\log n\log\log n})}$.
\end{theorem}

Given \Cref{thm:expander-routing}, we now describe an algorithm for \Cref{thm:expander-main}. Let $\tilde\epsilon=\epsilon/(k\cdot \textup{poly}(1/\phi)\cdot2^{O(\sqrt{\log n\log\log n})})$, where the $\textup{poly}(1/\phi)\cdot2^{O(\sqrt{\log n\log\log n})}$ factors come from \Cref{thm:expander-routing}. We run \Cref{thm:main-multi} on $G$, error parameter $\tilde\epsilon$, and the $k$-commodity demand $b\in\mathbb{R}^{V\times[k]}$ where $b_{j}=d_{j}(\mathbbm1_{s_{j}}-\mathbbm1_{t_{j}})$ for each $j\in[k]$. Then, \Cref{thm:main-multi} runs in $O(\tilde\epsilon^{-2}\log n \cdot(\left\lVert b\right\rVert_0+\tilde\epsilon^{-1}\left\lVert b\right\rVert_1))=O(\tilde\epsilon^{-2}\log n\cdot(k+\tilde\epsilon^{-1}D))$ time and returns a $k$-commodity flow $f$ where $f_{1},\ldots,f_{k}$ route demands $b_{1}-r_{1},\ldots,b_{k}-r_{k}\in\mathbb{R}^{V}$ for residual demands $r_{1},\ldots,r_{k}\in\mathbb{R}^{V}$ satisfying $|r_{j}(v)|\le\tilde\epsilon\deg(v)$ for all $v\in V,\,j\in[k]$.

We now wish to route the $k$-commodity (residual)\ demand $r$ defined by $r_{1},\ldots,r_{k}$ with total congestion of $k\tilde\epsilon\cdot\textup{poly}(1/\phi)\cdot2^{O(\sqrt{\log n\log\log n})}=\epsilon$ by calling \Cref{thm:expander-routing} \emph{once}. For each $j\in[k]$, arbitrarily decompose the demand $r_{j}$ into weighted pairs $(s_{j,1},t_{j,1},d_{j,1}),\ldots,(s_{j,\ell_{j}},t_{j,\ell_{j}},d_{j,\ell_{j}})$ such that $\sum_{i\in[\ell_{j}]}d_{j,\ell_{j}}\cdot(\mathbbm1_{s_{j,i}}-\mathbbm1_{t_{j,i}})=r_{j}$. A simple greedy algorithm achieves $\ell_{j}\le\left\lVert r_{j}\right\rVert _{0}$. Call \Cref{thm:expander-routing} on the weighted pairs $(s_{j,1},t_{j,1},d_{j,1}),\ldots,(s_{j,\ell_{j}},t_{j,\ell_{j}},d_{j,\ell_{j}})$ concatenated over all $j\in[k]$ and parameter $\alpha=k\tilde\epsilon$. Let $\ell=\ell_{1}+\dots+\ell_{k}$. Thus, \Cref{thm:expander-routing} returns an $\ell$-commodity flow with congestion $k\tilde\epsilon\cdot\textup{poly}(1/\phi)\cdot2^{O(\sqrt{\log n\log\log n})}=\epsilon$. For each $j\in[k]$, group the $\ell_{j}$ relevant flows into a single flow $\tilde{f}_{j}$ routing demand $r_{j}$. Observe that the flow $\hat{f}_{j}=f_{j}+\tilde{f}_{j}$ exactly routes demand $b_{j}$. Therefore, the $k$-commodity flow $\hat{f}$ defined by $\hat{f}_{1},\dots,\hat{f}_{k}$ precisely routes our $k$-commodity demand $b$ with total congestion $1+\epsilon$ as desired.

For running time, observe that the flow $f$ from \Cref{thm:main-multi} has support size (i.e., number of nonzero entries) at most $O(\tilde\epsilon^{-2}\log n\cdot(k+\tilde\epsilon^{-1}D))$ since that is also the running time. Constructing the $k$-commodity (residual)\ demand $r$ can be done in time linear in the support size of $f$, and the number of commodities in the call to \Cref{thm:expander-routing} has the same bound, i.e., $\ell=O(\tilde\epsilon^{-2}\log n\cdot(k+\tilde\epsilon^{-1}D))$. It follows, by \Cref{thm:expander-routing}, that the overall algorithm runs in time
\begin{align*}
&(m+O(\tilde\epsilon^{-2}\log n\cdot(k+\tilde\epsilon^{-1}D)))\cdot\textup{poly}(1/\phi)\cdot2^{O(\sqrt{\log n\log\log n})}
\\&=(m+\tilde\epsilon^{-2}k+\tilde\epsilon^{-3}D)\cdot\textup{poly}(1/\phi)\cdot2^{O(\sqrt{\log n\log\log n})}
\\&=(m+\epsilon^{-2}k^3+\epsilon^{-3}k^3D)\cdot\textup{poly}(1/\phi)\cdot2^{O(\sqrt{\log n\log\log n})}
\\&=(m+\epsilon^{-3}k^3D)\cdot\textup{poly}(1/\phi)\cdot2^{O(\sqrt{\log n\log\log n})},
\end{align*}
concluding the proof of \Cref{thm:expander-main}.